\documentclass[conference,letterpaper]{IEEEtran}
\usepackage[utf8]{inputenc} 
\usepackage[T1]{fontenc}
\usepackage{url}
\usepackage{ifthen}
\usepackage{cite}
\usepackage{amssymb}
\usepackage[cmex10]{amsmath} 
\usepackage{amsthm}
\usepackage{bm}
\usepackage{algorithm}
\usepackage{algpseudocode}
\usepackage{mathtools}

\newtheorem{theorem}{Theorem}
\newtheorem{proposition}{Proposition}
\newtheorem{corollary}{Corollary}
\theoremstyle{remark}
\newtheorem{remark}{Remark}

\usepackage[pdfversion=1.4,hidelinks, bookmarks=false]{hyperref}
\usepackage{cleveref}

\DeclareMathOperator{\E}{\mathbb{E}}

\begin{document}
\title{Importance Sampling for Event Discovery via Guesswork} 

\author{
  \IEEEauthorblockN{Asaf Cohen}
 \IEEEauthorblockA{School of Electrical and Computer Engineering\\
                    Ben-Gurion University of the Negev\\
                    Beer-Sheva\\
                    Email: coasaf@bgu.ac.il} }
\maketitle


\begin{abstract}
Traditional importance sampling (IS) is designed to estimate rare-event probabilities by minimizing estimator variance. However, many applications prioritize rapid discovery: the generation of a trajectory within a rare set \(A_n\). This requires a shift from ensemble-based estimation to a design principle focused on the hitting time \(\tau_{A_n}:= \inf\{t \ge 1: Y_t^n \in A_n\}\). 

We formalize a Quality of Discovery criterion based on the description length (surprisal) of trajectories under the nominal model \(p\). For i.i.d.\ models and type-defined rare sets \(\Gamma\), we show that minimizing this surprisal selects a target type that also minimizes the asymptotic nominal rank exponent $\lim \frac{1}{n} \log G_n(Y^n)$, where \(G_n(x^n)\) denotes the guesswork rank of \(x^n\) under \(p\). 

Thus, while classical IS targets the mass-dominating type \(Q_{\mathrm{IS}}^* \in \arg\min_{Q \in \Gamma} D(Q\|p)\), discovery optimality is achieved by \(Q_{\mathrm{GW}}^* \in \arg\min_{Q \in \Gamma} [H(Q) + D(Q\|p)]\). The resulting proposal \((Q_{\mathrm{GW}}^*)^{\otimes n}\) achieves subexponential hitting time, while its typical discoveries are guesswork-optimal: the least-surprising type under the nominal model is also one that appears earliest in the nominal probability-ordered search. This yields a simple rule for discovery-based IS: make the least-surprising rare trajectory typical. We further demonstrate that under budgetary constraints, the guesswork exponent serves as a lexicographic tie-breaker when the hitting-time minimizer is not unique. 
\end{abstract}

\section{Introduction}

Importance sampling (IS) is a powerful variance-reduction technique traditionally used to estimate the probability of rare events. By introducing an alternative sampling distribution $q_n$, IS allows for the efficient representation of the probability mass of a rare set $A_n$, which would otherwise require an exponential number of samples to observe under a nominal distribution $p_n$. In many modern applications, however, the primary objective shifts from probability estimation to \emph{rapid discovery}: the generation of a specific trajectory within a rare set as quickly as possible.

This shift to a discovery-based goal requires a design principle focused on the \emph{hitting time} $\tau_{A_n} := \inf\{t \ge 1 : Y_t^n \in A_n\}$, where $J_{\text{hit}}(q_n) := \limsup_{n\to\infty} \frac{1}{n} \log \mathbb{E}_{q_n}[\tau_{A_n}]$ is the hitting-time exponent. Rapid discovery necessitates a different proposal distribution $Q^*$, centered on discovery speed rather than variance minimization. We formalize this as the \emph{Quality of Discovery}, which seeks to hit the target set in subexponential time (ensuring $J_{\text{hit}}(q_n) = 0$) while optimizing the qualitative nature of typical discovered trajectories.

Searching for a distribution that efficiently hits a target connects importance sampling to the field of \emph{guesswork}. Guesswork centers on the search effort required to identify a realization of a random vector through sequential queries. The guesswork for a sequence $x^n$, $G_n(x^n)$, measures the number of queries required to identify $x^n$ when searching in decreasing order of probability. The guesswork exponent is the exponential rate at which $G_n(x^n)$ grows. Notably, this guesswork exponent is achievable not only through such a deterministic search but also through randomized guessing strategies, where samples are drawn independently according to a specified proposal distribution \cite{SHBCM18}. This raises a connection between importance sampling and guesswork.

Thus, to refine our choice among potential discovery proposals, we introduce a \emph{minimal description length} criterion. This approach seeks the trajectory that is most consistent with the nominal model: the "least surprising" representative of the rare set. We prove that minimizing nominal description length (surprisal) selects a type that also minimizes the guesswork exponent. For i.i.d.\ models and type-defined rare sets $\Gamma$, this leads to the concrete optimization
\begin{multline}\label{our optimization problem}
Q_{\mathrm{GW}}^*
\in
\arg\min_{Q\in\Gamma}
\mathbb{E}_{Q^{\otimes n}}
\left[
-\frac{1}{n}\log p^{\otimes n}(Y^n)
\right]
\\=
\arg\min_{Q\in\Gamma}
\bigl[H(Q)+D(Q\|p)\bigr].
\end{multline}
Since every $Q\in\Gamma$ yields subexponential hitting time, the optimization selects, among fast-discovery proposals, a type whose typical trajectories are least surprising under the nominal model and are also guesswork-optimal. Equivalently, the proposal $Q_{\mathrm{GW}}^*$ makes a least-surprising rare type typical.

The contributions of this paper are summarized as follows. We formalize \emph{Quality of Discovery} as a criterion for importance sampling that emphasizes the properties of individual discovered trajectories rather than only the probability mass of the rare set. We show that nominal description-length minimization yields guesswork-optimal discovery, and derive the corresponding proposal $Q_{\mathrm{GW}}^*$ for i.i.d.\ models. Finally, under constraints on the admissible change of measure, we show that the guesswork exponent provides a \emph{lexicographic tie-breaker} among proposals with the same optimal hitting-time exponent.

\subsection{Importance Sampling}

The references most relevant to this work are at the intersection of classical rare-event importance sampling, modern path-space sampling, and recent "make rare typical" constructions. While these distinct lines of research focus on optimizing estimation accuracy (minimizing variance), our approach shifts the focus toward optimizing discovery speed and the qualitative representativeness of trajectories through a guesswork criterion.

Early foundations in importance sampling were primarily concerned with sequential analysis and stopping times. \cite{siegmund1976importance} provided a classical treatment in this area; however, this work tunes the change of measure for accurate probability estimation rather than rapid discovery. \cite{bucklew2004introduction} is a standard reference for exponential tilting and large-deviations-based proposal design. While we utilize similar change-of-measure language, we replace this variance-efficiency metric with expected hitting time.

Similarly, \cite{SadowskyBucklew1990LDMC} formalized asymptotically efficient Monte Carlo estimation. While this work also studied exponential rates, our work centers on a search exponent driven by guesswork rather than a second-moment efficiency exponent. \cite{Bucklew2005ConditionalIS} further explored conditional estimators from an information-theoretic viewpoint. We maintain this proposal-design emphasis but, again, evaluates proposals based on the mass they place on the target set rather than estimator variance.

The connection between importance sampling, large deviations, and control was established in \cite{DupuisWang2004}, viewing optimal proposals as a game between change of measure and large deviation. Similarly, \cite{blanchet2008efficient} developed state-dependent schemes for heavy-tailed random walks. Although they target hitting events, their optimization objective remains estimator efficiency rather than direct hitting-time cost.

For alternative algorithms, trajectory biasing and making the rare typical, \cite{cerou2007adaptive} enhanced adaptive multilevel splitting, where one multiplies certain (favorable) trajectories and excludes others. Still, the focus is on estimation of (very) rare events. More recent literature has moved toward forcing rare transitions to occur more frequently. \cite{HartmannSchutte2012NoneqForcing} modified dynamics through optimal nonequilibrium forcing to sample rare transitions efficiently. While conceptually close to our goal of making rare events visible, we derive the preferred bias from nominal-probability ranking and guesswork. The recent review \cite{SinghDasLimmer2025AnnualReviewVPS} considered the broader landscape of variational path-sampling, also searching for a biased distribution, which minimizes the statistical distance to the true rare-event distribution, yet making the improbable events appear typical in simulation.

This philosophy is shared in \cite{CarolloGarrahanLesanovskyPerezEspigares2018QuantumDoob}, where an explicit construction modifies dynamics under which rare trajectories become typical. Our work expresses this same intuition within an information-theoretic, guesswork-based framework for discrete sequence models. In a similar vein \cite{DuPlainerBrekelmansDuanNoeGomesAspuruGuzikNeklyudov2024DoobsLagrangian} used variational methods to approximate Doob-transformed dynamics for rare path sampling. While their target is conditioned path sampling, ours is hitting-time optimization for discovery.

Modern computational approaches have also incorporated machine learning. \cite{RoseMairGarrahan2021RLRareTraj} utilized reinforcement learning to learn proposal dynamics. Our approach differs by changing the underlying criterion being learned from estimator quality to discovery speed. 

\subsection{Guesswork}
The study of guesswork centers on the search effort required to identify a realization of a random $n$-vector $X^n$ through sequential queries. It is well established that the optimal strategy for minimizing guesswork moments involves querying candidates in descending order of their likelihood. This framework has its roots in early research on sequential decoding \cite{wozencraft1957sequential,Arikan96} and the discovery of codewords under specific structural constraints \cite{PfisterSullivan}.

A cornerstone of the field is the work by Arıkan, which proved that the asymptotic growth rate of guesswork moments is governed by the source's Rényi entropy \cite{Arikan96}. While much focus has been on these large deviations and asymptotic limits \cite{christiansen2013guessing,christiansen2013guesswork}, recent investigations have produced non-asymptotic converse results that clarify the behavior at finite block lengths \cite{courtade14}.
Information-theoretic perspectives further link guesswork to one-shot or fixed-to-variable source coding without prefix constraints \cite{Szpankowski11F2V,Kontoyiannis14optimal_lossless,Courtade_Verdu_14,Kosut_universal_F2V_17}. In these models, the objective of ordering sequences by probability of appearance is functionally equivalent to the guesswork optimization problem.
More sophisticated models incorporate additional parameters, such as side information, which can significantly lower the required search effort \cite{Arikan96,4036408,christiansen2013guessing,salamatian2017centralized,Sason_Verdu_18}. Furthermore, to address constraints like memory limitations or lack of agent synchronization, researchers have explored randomized guessing strategies \cite{SHBCM18,Salamatian2020CentralizedvsDecentralized,Merhav2020RandomizedGuessing,Cohen2022GuessingWithDistortion}. In these setups, guesses are drawn independently according to a specified distribution rather than following a deterministic list.

Finally, note that the i.i.d.\ setting provides a clear conceptual proof-of-concept rather than a final computational goal. It allows us to separate two fundamentally different questions in the simplest possible setting: which part of the rare set carries most of its probability mass, and which rare trajectories are easiest to discover in the nominal probability ordering. While this distinction is often hidden in complex models, where rare events are defined by intricate paths and proposals are restricted to specific dynamics, the i.i.d.\ analysis identifies the correct objective. We view this as a fundamental blueprint: when mathematical solutions are out of reach in high-dimensional systems, this framework guides us to choose proposals that target the most "discoverable" trajectories rather than just the most probable regions. 
\section{Preliminaries}

\subsection{Notation}

Throughout, \(\mathcal{X}\) denotes a finite alphabet, and \(\mathcal{X}^n\) the set of all length-\(n\) sequences over \(\mathcal{X}\). Elements of \(\mathcal{X}\) are denoted by \(x_i\), while sequences are written as \(x^n=(x_1,\dots,x_n)\in\mathcal{X}^n\).

Random variables are denoted by uppercase letters, such as \(X\), and their realizations by lowercase letters, such as \(x\). Random sequences are written as \(X^n=(X_1,\dots,X_n)\), where each \(X_i\) takes values in \(\mathcal{X}\). Probability distributions on \(\mathcal{X}\) are denoted by \(P,Q,p,q\), according to context.

For a sequence \(x^n\in\mathcal{X}^n\), its empirical distribution, or \emph{type}, is a probability mass function, \(\hat P_{x^n}\) on \(\mathcal{X}\). For a type \(Q\) on \(\mathcal{X}\), the associated \emph{type class} is $T(Q) := \{x^n\in\mathcal{X}^n : \hat P_{x^n}=Q\}$. We write \(\mathcal{P}_n(\mathcal{X})\) for the set of all \(n\)-types on \(\mathcal{X}\).

Throughout this paper, we use the notation $a_n \doteq b_n$ to denote that $a_n$ is exponentially equivalent to $b_n$, namely,
\begin{equation*}
    \lim_{n \to \infty} \frac{1}{n} \log a_n = \lim_{n \to \infty} \frac{1}{n} \log b_n.
\end{equation*}

\subsection{Classical Importance Sampling}
Let \(p_n\) be a probability distribution on \(\mathcal{X}^n\). In many applications, one is interested in a \emph{rare event} \(A_n \subseteq \mathcal{X}^n\), whose probability under \(p_n\) is exponentially small. Direct Monte Carlo estimation of \(p_n(A_n)\) is inefficient when \(A_n\) is rare, since the number of samples required to observe even a single occurrence scales as \(1/p_n(A_n)\), which is typically exponential in \(n\). Importance sampling addresses this problem by introducing an alternative sampling distribution \(q_n\) on \(\mathcal{X}^n\), such that samples are drawn according to \(q_n\), rather than \(p_n\). The probability \(p_n(A_n)\) can then be written as
\[
p_n(A_n)
=
\mathbb{E}_{q_n}\left[
\mathbf{1}\{X^n \in A_n\}
\cdot
\frac{p_n(X^n)}{q_n(X^n)}
\right],
\qquad X^n \sim q_n.
\]
This yields the unbiased estimator
\[
\hat{p}_N(A_n)
=
\frac{1}{N}
\sum_{i=1}^N
\mathbf{1}\{X_i^n \in A_n\}
\cdot
\frac{p_n(X_i^n)}{q_n(X_i^n)},
\qquad X_i^n \stackrel{i.i.d.}{\sim} q_n.
\]

The key design problem in importance sampling is therefore: \emph{How should one choose \(q_n\) so that the estimator has low variance?}

\subsection{Optimal and Near-Optimal Choices}

The optimal (zero-variance) choice is $q_n^\star(x^n)=\frac{p_n(x^n)}{p_n(A_n)}
\mathbf{1}\{x^n \in A_n\}$, which yields a deterministic estimator. Yet, this choice requires knowledge of \(p_n(A_n)\), and identification of each trajectory in $A_n$. Practical schemes therefore aim to approximate \(q_n^\star\), typically by \emph{tilting} the nominal distribution \(p_n\) toward the set \(A_n\). Thus, classical IS is concerned with estimating the probability \(p_n(A_n)\), designing \(q_n\) so that samples from \(q_n\) efficiently represent the mass of \(A_n\), and controlling the variance of the resulting estimator. In particular, the quality of \(q_n\) is judged by how well it captures the \emph{total probability mass} of \(A_n\).

\subsection{Guesswork}
Let \(\{x^n_{(k)}\}\) denote the ordering of \(\mathcal{X}^n\) such that $p_n(x^n_{(1)}) \ge p_n(x^n_{(2)}) \ge \cdots$. The guesswork of a sequence \(x^n\) is defined as $G_n(x^n) := \#\{y^n \in \mathcal{X}^n : p_n(y^n) \ge p_n(x^n)\}$.
Equivalently, \(G_n(x^n)\) is the rank of \(x^n\) in decreasing \(p_n\)-order. The quantity \(G_n(x^n)\) depends only on the probability level \(p_n(x^n)\), and induces a monotone ranking of \(\mathcal{X}^n\). Thus, guesswork defines a geometry on path space in which sequences with larger \(p_n(x^n)\) are easier to discover. 

For a set \(A_n \subseteq \mathcal{X}^n\), define $G_n(A_n) := \min_{x^n \in A_n} G_n(x^n)$. Then 
\[
G_n(A_n) = \#\Bigl\{y^n : p_n(y^n) \ge \max_{x^n \in A_n} p_n(x^n)\Bigr\}.
\]
For a set \(A_n\), the relevant quantity is not \(p_n(A_n)\), but rather $\max_{x^n \in A_n} p_n(x^n)$, i.e., the most likely element of \(A_n\). Accordingly, \(G_n(A_n)\) measures the discovery complexity of \(A_n\) through its most accessible trajectories.
\section{Guesswork and Importance Sampling}

\subsection{From Hitting Time to Minimal Surprisal Discovery}

In many settings, the goal of importance sampling is not merely to estimate $p_n(A_n)$, but to \emph{generate} samples in $A_n$, i.e., to actually observe rare trajectories. This shift from estimation to generation leads to a different performance metric: the \emph{hitting time}
\[
\tau_{A_n}
:=
\inf\{t \ge 1 : Y_t^n \in A_n\},
\qquad Y_t^n \stackrel{i.i.d.}{\sim} q_n,
\]
for which the expected value is $\mathbb{E}_{q_n}[\tau_{A_n}] = \frac{1}{q_n(A_n)}$.
Thus, instead of minimizing estimator variance, one is led to a design problem focused on discovery speed: choose $q_n$ to maximize $q_n(A_n)$, or equivalently, to ensure the hitting-time exponent $J_{\mathrm{hit}}(q_n) := \limsup_{n\to\infty} \frac{1}{n} \log \mathbb{E}_{q_n}[\tau_{A_n}]$ is minimized.

However, in the asymptotic regime, the constraint $J_{\mathrm{hit}}(q_n) = 0$ is often easy to satisfy. For any rare set $A_n$ defined by a set of types $\Gamma$, any proposal $q_n = Q^{\otimes n}$ with $Q \in \Gamma$ ensures that the event is typical, resulting in discovery in subexponential time. Since many proposals can satisfy this requirement, we require a secondary criterion to ensure the \emph{quality} of the discovered trajectory. Specifically, we propose to select, among all proposals that achieve fast discovery, the one that minimizes the individual description length (surprisal) of the discovered trajectory under the nominal model $p_n$, as given in \eqref{our optimization problem}. This ensures that the discovery is a "normal" representative of the rare set rather than a statistical outlier. 

\paragraph{Connection to Guesswork}
For i.i.d.\ nominal models $p_n = P^{\otimes n}$ and product-form proposals $q_n = Q^{\otimes n}$, the objective function has a clean information-theoretic limit. By the law of large numbers for types, a trajectory $Y^n$ sampled from $q_n$ will have a type close to $Q$. For $Y^n \sim q_n$ the expected normalized nominal surprisal is exactly the cross-entropy between \(Q\) and \(p\):
\begin{equation}
    \mathbb{E}_{q_n} \left[ -\frac{1}{n} \log p_n(Y^n) \right] = H(Q) + D(Q\|p).
\end{equation}
This expression determines the nominal probability level of type $Q$, and hence governs the exponent of in $G_n(x^n) = \#\{y^n: p_n(y^n) \ge p_n(x^n)\}$ in \Cref{prop:pathwise_guesswork} below. In particular, a type that minimizes the nominal surprisal $H(Q)+D(Q\|p)$ also minimizes the guesswork exponent, and therefore appears among the earliest rare types in a deterministic search ordered by the nominal model. Thus, $Q_{\mathrm{GW}}^* \in \arg\min_{Q\in\Gamma}[H(Q)+D(Q\|p)]$ 
identifies a guesswork-optimal target type, and the corresponding product proposal $(Q_{\mathrm{GW}}^*)^{\otimes n}$ provides the connection between randomized sampling and systematic search.

\subsection{Pathwise guesswork exponent for a fixed sequence}
We first establish $g(Q)$ below as the pathwise quantity that connects the proposal-design problem in importance sampling to the deterministic search in guesswork. While this pathwise analysis is the primary focus for importance sampling, we also demonstrate how this characterization recovers the classical statistical limits of guesswork. Specifically, when the target sequence $X^n$ is random and distributed according to $p^{\otimes n}$, the average of the pathwise exponent $g(Q_{X^n})$ coincides with Ar\i{}kan's guesswork exponent. 
\begin{proposition}\label{prop:pathwise_guesswork}
Assume \(p_n=p^{\otimes n}\). Let \(x^n\in\mathcal X^n\) have type \(Q_n=\hat P_{x^n}\), and suppose \(Q_n\to Q\in\mathcal P(\mathcal X)\). Then
\[
G_n(x^n)
=
\sum_{\substack{Q'\in\mathcal P_n(\mathcal X):\\
H(Q')+D(Q'\|p)\le H(Q_n)+D(Q_n\|p)}}
|T(Q')|.
\]
Consequently, $\lim_{n\to\infty}\frac{1}{n}\log G_n(x^n)=g(Q)$,
where
\[
g(Q):=
\sup_{Q'\in\mathcal P(\mathcal X):\, H(Q')+D(Q'\|p)\le H(Q)+D(Q\|p)} H(Q').
\]
\end{proposition}
\begin{proof}
All sequences $y^n\in T(Q')$ have nominal probability $p_n(y^n)=\exp\bigl(-n[H(Q')+D(Q'\|p)]\bigr)$. Hence the guesswork rank of \(x^n\) counts those type classes whose nominal probability is at least \(p_n(x^n)\), which gives the first identity. Using the type-class estimate \(|T(Q')|\doteq e^{nH(Q')}\) and the fact that the number of \(n\)-types grows polynomially in \(n\), the exponential growth rate is given by the largest entropy among types satisfying the same nominal-probability constraint. Continuity of \(H(\cdot)\) and \(D(\cdot\|p)\) then yields the limit along \(Q_n\to Q\).
\end{proof}

\begin{remark}
Let $\ell(x)=-\log p(x)$. For every $R$, the quantity
$H(R)+D(R\|p)=\sum_x R(x)\ell(x)=:C(R)$ is \emph{linear} in $R$. Hence, by Proposition~1,
\[
  g(Q)\;=\;\max\Big\{H(R):\ \textstyle\sum_x R(x)\ell(x)\;\le\; C(Q)\Big\}
\]
is the maximum of a strictly concave function over a half-space, i.e.\ a maximum-entropy
problem under a single linear (mean-surprisal) constraint. Its Lagrangian solution is a
\emph{tilt of $p$}:
\[
  R_\lambda(x)=\frac{p(x)^{\lambda}}{\sum_{a}p(a)^{\lambda}},\qquad \lambda\ge 0,
\]
where $\lambda$ is chosen so that $C(R_\lambda)=C(Q)$ ($C(R_\lambda)$ decreases from the value
at the uniform law, $\lambda=0$, to $\min_x \ell(x)$, so $\lambda$ is unique); if $C(Q)$ is at least the
value at the uniform law, the constraint is inactive and $g(Q)=\log|\mathcal X|$. Thus
\[
  g(Q)=H(R_{\lambda(Q)}) \;\ge\; H(Q),
\]
since $Q$ itself is feasible, with equality if and only if $Q$ is itself a tilt of $p$.
In general $g(Q)\neq H(Q)$: the rank of a sequence of type $Q$ counts \emph{all} sequences
that are at least as probable, and these include the larger-entropy types lying on the same
level set $C(\cdot)=C(Q)$. 

\emph{Binary case.} For $\mathcal X=\{0,1\}$ with $p(1)>\tfrac12$, the tilts
$R_\lambda(1)=p(1)^\lambda/(p(1)^\lambda+p(0)^\lambda)$ sweep the whole interval
$[\tfrac12,1)$ as $\lambda\in[0,\infty)$, so every $Q$ with $q_1\ge\tfrac12$ is its own tilt and
$g(Q)=H(Q)$. For $q_1<\tfrac12$ the constraint is inactive and $g(Q)=1$ bit, whence
\[
  g(Q)=H\big(\max\{q_1,\tfrac12\}\big).
\]
\end{remark}

\begin{corollary}\label{guesswork for set}
Assume $p_n = p^{\otimes n}$ and let $\Gamma \subseteq \mathcal{P}(\mathcal{X})$ be a set of types. For the target set $A_n(\Gamma) = \{y^n \in \mathcal{X}^n : \hat{P}_{y^n} \in \Gamma\}$, the set-guesswork satisfies $G_n(A_n(\Gamma)) \doteq \exp\left( n \inf_{Q \in \Gamma} g(Q) \right)$.
\end{corollary}
\if{0}
Note, however, that $g(Q) \ge H(Q)$ and the inequality may be strict.\footnote{For example, consider a binary alphabet with a nominal distribution $p = (0.2, 0.8)$ and a type $Q = (0.6, 0.4)$. In this case, there exist other types $Q'$ such that $H(Q') + D(Q'\|p) \le H(Q) + D(Q\|p)$ but $H(Q') > H(Q)$, leading to $g(Q) > H(Q)$.} Hence, when taking the sum over all types, the objective becomes $\rho g(Q) - D(Q\|p)$, which does not recover Arıkan's exponent directly. Recovery still works because the global optimization balances the benefit of a larger pathwise exponent with the corresponding divergence cost: any type $Q$ that yields a high $g(Q)$ must also have a high nominal cost, which exactly offsets the gain in the global supremum. 
\fi
\begin{proof}
By definition, the guesswork of a set is determined by its most likely element:
\begin{equation}
    G_n(A_n(\Gamma)) = \min_{x^n \in A_n(\Gamma)} G_n(x^n).
\end{equation}
From the pathwise characterization of guesswork for a single sequence of type $Q$ in \Cref{prop:pathwise_guesswork}, we obtain
\begin{align}
    G_n(A_n(\Gamma)) &= \min_{Q \in \Gamma \cap \mathcal{P}_n(\mathcal{X})} G_n(x^n : \hat{P}_{x^n} = Q) 
    \\
    & \doteq \exp\left( n \inf_{Q \in \Gamma} g(Q) \right).
\end{align}
\end{proof}
For completeness, we show how this pathwise characterization recovers Ar\i{}kan's classical guesswork exponent when the target sequence is random and distributed according to \(p^{\otimes n}\). The proof is technical and appears in Appendix \ref{app:arikan}.
\begin{proposition}[Recovering Arikan's Exponent]
\label{prop:arikan_average}
Let $X^n$ be an i.i.d.\ sequence drawn according to $p^{\otimes n}$. Let $g(Q)$ be the pathwise guesswork exponent for a sequence of type $Q$. Then, the asymptotic growth of the $\rho$-th moment of guesswork satisfies:
\begin{equation*}
\lim_{n \to \infty} \frac{1}{n} \log \mathbb{E}_{p_n} [G_n(X^n)^\rho] = \sup_{Q \in \mathcal{P}(\mathcal{X})} \left[ \rho g(Q) - D(Q \| p) \right].
\end{equation*}
Furthermore, this optimization recovers Arikan's Exponent: $\sup_{Q \in \mathcal{P}(\mathcal{X})} \left[ \rho g(Q) - D(Q \| p) \right] = \rho H_{\frac{1}{1+\rho}}(p)$, where $H_\beta(p) = \frac{1}{1-\beta} \log \sum_{a \in \mathcal{X}} p(a)^\beta$ is the R\'enyi entropy of order $\beta$.
\end{proposition}



\subsection{Guesswork-Guided Importance Sampling Design}
We can now state the optimality of guesswork-guided sampling.
\begin{theorem}\label{main theorem unconstrained}
Assume a nominal i.i.d.\ model $p_n = p^{\otimes n}$ and a rare event $A_n(\Gamma)$ defined by a set of target types $\Gamma \subseteq \mathcal{P}(\mathcal{X})$. Let $q_n^* = (Q_{\mathrm{GW}}^*)^{\otimes n}$ be the proposal distribution where $Q_{\mathrm{GW}}^* \in \arg\min_{Q \in \Gamma} [H(Q) + D(Q\|p)]$. Assume that $\Gamma$ is regular in the sense that $Q_{\rm GW}^*$ can be approached by $n$-types in $\Gamma$. Then, the discovery process defined in Algorithm 1 below achieves the following:

\textbf{Subexponential Hitting Time:} The hitting-time exponent satisfies $J_{\mathrm{hit}}(q_n^*) := \limsup_{n\to\infty} \frac{1}{n} \log \mathbb{E}_{q_n^*}[\tau_{A_n}] = 0$. Consequently, the expected number of samples required for discovery is sub-exponential in $n$.

\textbf{Asymptotic Rank Optimality:} If $Y^n$ denotes the first sample generated by $q_n^*$ that belongs to $A_n(\Gamma)$, the nominal rank $G_n(Y^n)$ of the discovered trajectory $Y^n$ satisfies $\lim_{n\to\infty} \frac{1}{n} \log G_n(Y^n) = g(Q_{\mathrm{GW}}^*)$ in probability. Furthermore, this value coincides with the set-guesswork exponent $\inf_{Q \in \Gamma} g(Q)$, which is the fundamental lower bound on the search depth required to encounter any element of $A_n(\Gamma)$.
\end{theorem}
\begin{proof}
The hitting time $\tau_{A_n}$ for a sequence of i.i.d.\ trials $Y_t^n \sim q_n^*$ follows a geometric distribution with success probability $\pi_n = q_n^*(A_n(\Gamma))$. The expected hitting time is given by $\mathbb{E}_{q_n^*}[\tau_{A_n}] = 1/\pi_n$. By the property of the optimal proposal $q_n^* = (Q_{\mathrm{GW}}^*)^{\otimes n}$, where $Q_{\mathrm{GW}}^* \in \Gamma$, the law of large numbers for types implies that the event $A_n(\Gamma)$ is typical under $q_n^*$. Specifically, by Sanov's Theorem, since $Q_{\mathrm{GW}}^* \in \Gamma$,
\[
-\lim_{n\to\infty} \frac{1}{n} \log q_n^*(A_n(\Gamma)) = \inf_{Q \in \Gamma} D(Q \| Q_{\mathrm{GW}}^*) = 0.
\]
Therefore, the hitting-time exponent is:
\[
J_{\mathrm{hit}}(q_n^*) = \limsup_{n\to\infty} \frac{1}{n} \log \frac{1}{q_n^*(A_n(\Gamma))} = 0.
\]

Let $Y^n$ be the trajectory returned by Algorithm 1. The algorithm repeatedly draws samples from $(Q_{\mathrm{GW}}^*)^{\otimes n}$ until a sample satisfying $\hat P_{Y^n}\in\Gamma$ is obtained. Since $Q_{\mathrm{GW}}^*\in\Gamma$, and under the regularity assumption on $\Gamma$, the empirical type of the returned trajectory concentrates around $Q_{\mathrm{GW}}^*$. Hence, $\hat P_{Y^n}\xrightarrow{P}Q_{\mathrm{GW}}^*$. However, by Proposition \ref{prop:pathwise_guesswork}, for every series of sequences $x^n$ with types $Q_n\to Q$ we have $\lim_{n\to\infty}\frac{1}{n}\log G_n(x^n)=g(Q)$. Hence, the normalized nominal rank of $Y^n$ satisfies:
\[
\frac{1}{n} \log G_n(Y^n) \xrightarrow{P} g(Q_{\mathrm{GW}}^*) = \min_{Q \in \Gamma} g(Q).
\]
To show the last equality, note first that $g(Q)$ is a non-decreasing function of $C(Q):=H(Q)+D(Q\|p)$. Thus, by the definition of $g(Q)$, any minimizer of $C(Q)$ produces a minimizer of $g(Q)$ and we have $\arg\min_{Q\in \Gamma} C(Q) \subseteq \arg\min_{Q\in \Gamma} g(Q)$. Hence, $Q_{\mathrm{GW}}^*$ minimizing $C(Q)$ satisfies $g(Q_{\mathrm{GW}}^*) = \min_{Q \in \Gamma} g(Q)$.
\end{proof}
\begin{algorithm}\label{our main alg}
\caption{Guesswork-Guided Event Discovery}
\label{alg:gw_discovery}
\begin{algorithmic}[1]
\State \textbf{Input:} Nominal distribution $p$, target set $\Gamma$.
\State \textbf{Output:} Discovered trajectory $Y^n \in A_n(\Gamma)$.

\Statex \Comment{\textit{Step 1: Solve the Discovery Optimization}}
\State Find $Q_{\mathrm{GW}}^* \in \arg\min_{Q \in \Gamma} [H(Q) + D(Q\|p)]$
\Statex \hspace{1em} \small (Note: $\lim_{n\to\infty} \frac{1}{n} \log G_n(A_n(\Gamma)) = g(Q_{\mathrm{GW}}^*)$)

\Statex \Comment{\textit{Step 2: Construct the Proposal}}
\State Define $q_n(y^n) = \prod_{i=1}^n Q_{\mathrm{GW}}^*(y_i)$

\Statex \Comment{\textit{Step 3: Importance Sampling for Discovery}}
\Repeat
    \State Sample $Y^n \sim (Q_{\mathrm{GW}}^*)^{\otimes n}$
\Until{$\hat{P}_{Y^n} \in \Gamma$}

\State \Return $Y^n$
\end{algorithmic}
\end{algorithm}
\begin{remark}[Simplicity of Finding $Q^*_{GW}$]
In the finite alphabet case $\mathcal{X} = \{1, \dots, |\mathcal{X}|\}$, the objective simplifies to $H(Q) + D(Q||p) = -\sum_{x \in \mathcal{X}} Q(x) \log p(x)$, which is linear in $Q$. Hence, if $\Gamma$ is a nonempty closed convex polytope defined by linear constraints (e.g., expectation constraints $\mathbb E_Q[f(X)]\ge\alpha$), an optimal solution can always be chosen among the extreme points of $\Gamma$. The coefficients $-\log p(x)$ favor assigning probability mass to symbols that are more likely under $p$, subject $\Gamma$. Geometrically, the optimizer lies on a boundary face of $\Gamma$ exposed in the direction $-\log p$.
\end{remark}

\subsection{Example: A Binary Source}
\label{subsec:binary-example}

Let $p=[0.8,0.2]$ and $\Gamma=\Gamma_L\cup\Gamma_R$, with $\Gamma_L=\{q_1\le 0.6\}$ and
$\Gamma_R=\{0.97\le q_1\le 0.99\}$, where $q_1$ is the fraction of ones.
Both components are rare under $p$. ($\Gamma_R$ is bounded so that the minimizer of
$H+D$ is $0.99$ and not the degenerate all-ones type.)

The two objectives are minimized over $\Gamma$ at different points (Fig.~\ref{fig:ex_obj}):
classical IS selects $Q_A=[0.6,0.4]$ (left component, closer to $p$ in KL), whereas
\eqref{our optimization problem} selects $Q_B=[0.99,0.01]$ (right component, less surprising).
For a binary $p$ with $p(1)>1/2$, the maximizer in Proposition~1 is $Q$ itself when
$q_1\ge 1/2$, and the uniform law otherwise, hence $g(Q)=H(\max\{q_1,\tfrac12\})$.
Table~\ref{tab:ex_values} lists the theoretical values; since $g$ is nondecreasing in
$H+D$ (Fig.~\ref{fig:ex_g}), $Q_B$ also minimizes the guesswork exponent.

\begin{table}[t]
\centering
\caption{Theoretical values (bits) for the binary example.}
\label{tab:ex_values}
\setlength{\tabcolsep}{3pt}\small
\begin{tabular}{c|cccc}
 & $D(Q\|p)$ & $H(Q)+D(Q\|p)$ & $g(Q)$ & $J_{\mathrm{hit}}$ \\ \hline
$Q_A=[0.6,0.4]$   & $0.151$ & $1.122$ & $0.971$ & $0$ \\
$Q_B=[0.99,0.01]$ & $0.261$ & $0.342$ & $0.081$ & $0$
\end{tabular}
\end{table}

\emph{Simulation.}
Algorithm~1 was run $2000$ times for each $n\in\{100,200,400,800,1600\}$ and each of
$Q_A,Q_B$, recording $\tau_{A_n}$, $-\frac1n\log p^{\otimes n}(Y^n)$ and
$\frac1n\log G_n(Y^n)$. All three depend on $Y^n$ only through the number of ones $k$:
a draw is simulated as $k\sim\mathrm{Bin}(n,q_1)$, and the rank is exact,
$G_n=\sum_{j\ge k}\binom{n}{j}$ (both were verified against literal sequences and
brute-force sorting for small $n$).

Both proposals discover $A_n$ in about two draws, versus $2^{0.151n}$ under $p$
(Fig.~\ref{fig:ex_hit}). The discovered sequences differ: the normalized surprisal
converges to $1.122$ bits for $Q_A$ and $0.342$ bits for $Q_B$ (Fig.~\ref{fig:ex_surp}),
and the rank exponent to $g(Q_A)=0.971$ and $g(Q_B)=0.081$ (not shown). Thus $Q_B$ finds,
in the same time, sequences that are $\approx 2^{0.89n}$ positions earlier in the nominal
list, essentially the most probable sequences of $A_n$.
Proposals outside $\Gamma$ (Theorem~2) are slower, with exponent $D(R^*_Q\|Q)$
(Fig.~\ref{fig:ex_sweep}).
Conversely, $Q_B$ is a poor estimator of $p^{\otimes n}(A_n)$: it essentially never
visits $\Gamma_L$, which carries most of the mass (Fig.~\ref{fig:ex_est}).
The two criteria thus answer different questions.

\begin{figure}[t]
\centering
\includegraphics[width=\columnwidth]{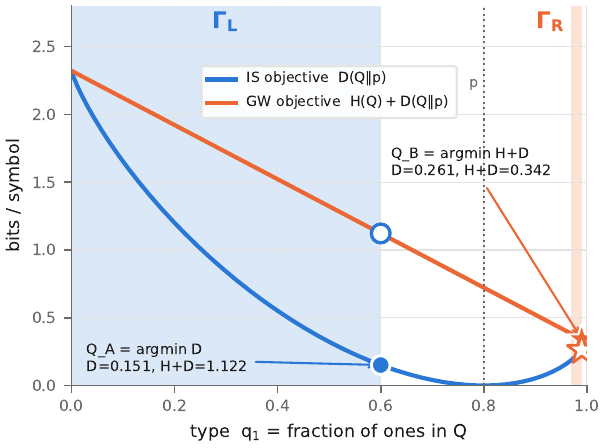}
\caption{$D(Q\|p)$ (blue) and $H(Q)+D(Q\|p)$ (orange) versus the type $q_1$ of $Q$;
$\Gamma_L,\Gamma_R$ shaded, dotted line at $p$. Classical IS selects $Q_A$ (circle),
guesswork selects $Q_B$ (star); open markers show the other objective at the same point.}
\label{fig:ex_obj}
\end{figure}

\begin{figure}[t]
\centering
\includegraphics[width=\columnwidth]{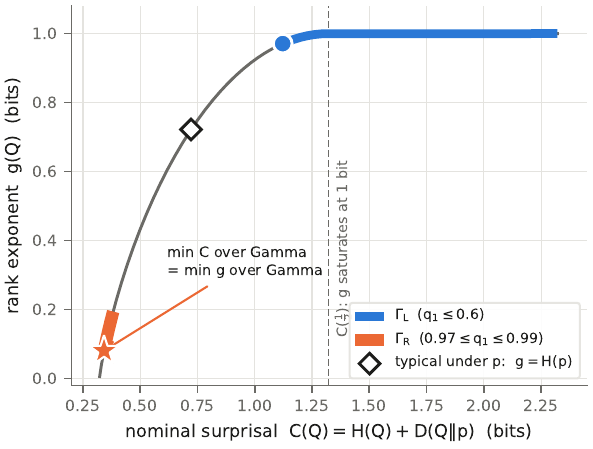}
\caption{Rank exponent $g(Q)$ versus $H(Q)+D(Q\|p)$ over all $q_1$; thick segments are
types in $\Gamma_L$ (blue) and $\Gamma_R$ (orange); diamond: $Q=p$. $g$ is nondecreasing
in $H+D$ and saturates at $1$ bit beyond $H+D=1.32$ (dashed), so $Q_B$ minimizes both.}
\label{fig:ex_g}
\end{figure}

\begin{figure}[t]
\centering
\includegraphics[width=\columnwidth]{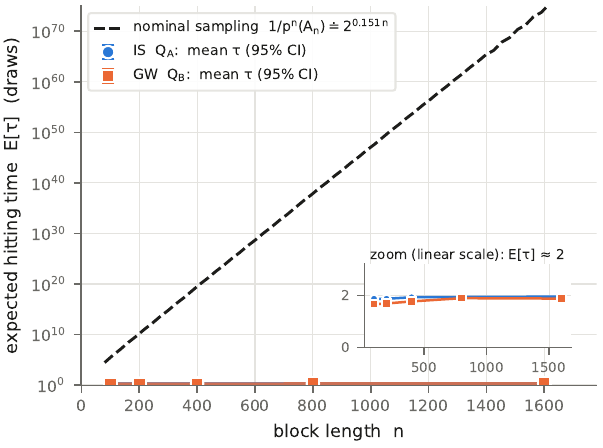}

\caption{Expected hitting time versus $n$. Dashed: sampling from $p$,
$1/p^{\otimes n}(A_n)\doteq 2^{0.151n}$. Markers: sample mean of $\tau_{A_n}$ over $2000$
independent runs of Algorithm~1 ($\pm1.96$ standard errors, smaller than the markers) for
$Q_A$ (blue) and $Q_B$ (orange); lines: exact $1/Q^{\otimes n}(A_n)$.
Inset: linear scale; $J_{\mathrm{hit}}=0$ for both.}
\label{fig:ex_hit}
\end{figure}

\begin{figure}[t]
\centering
\includegraphics[width=\columnwidth]{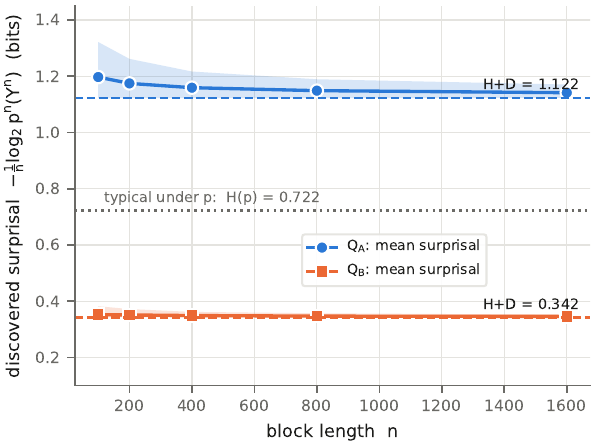}
\caption{Normalized surprisal $-\frac1n\log p^{\otimes n}(Y^n)$ of the discovered sequence
(mean and $5$--$95\%$ band over runs). Dashed: limits $H+D=1.122$ ($Q_A$) and $0.342$ ($Q_B$);
dotted: $H(p)=0.722$ (typical under $p$).}
\label{fig:ex_surp}
\end{figure}

\begin{figure}[t]
\centering
\includegraphics[width=\columnwidth]{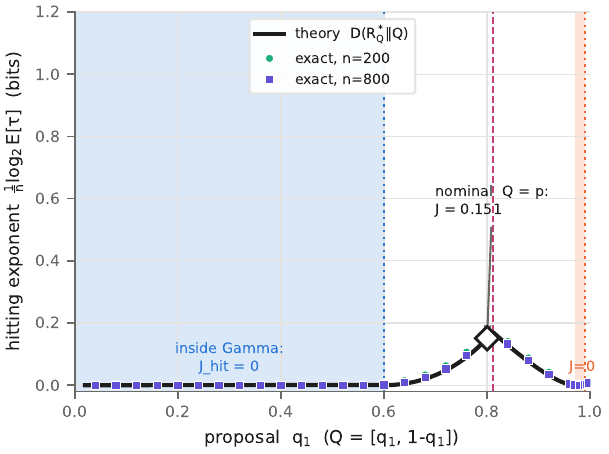}
\caption{Hitting exponent $\frac1n\log E[\tau_{A_n}]$ versus the proposal $Q=[q_1,1-q_1]$,
including $Q\notin\Gamma$. Solid: $D(R^*_Q\|Q)$ with $R^*_Q=\arg\min_{R\in\Gamma}D(R\|Q)$;
markers: exact values at $n=200,800$. It is zero inside $\Gamma$ (shaded); diamond:
$Q=p$ ($0.151$). Dotted lines: $Q_A,Q_B$; dashed: $q_1\approx0.81$, where $R^*_Q$ switches
from $\Gamma_L$ to $\Gamma_R$.}
\label{fig:ex_sweep}
\end{figure}

\begin{figure}[t]
\centering
\includegraphics[width=\columnwidth]{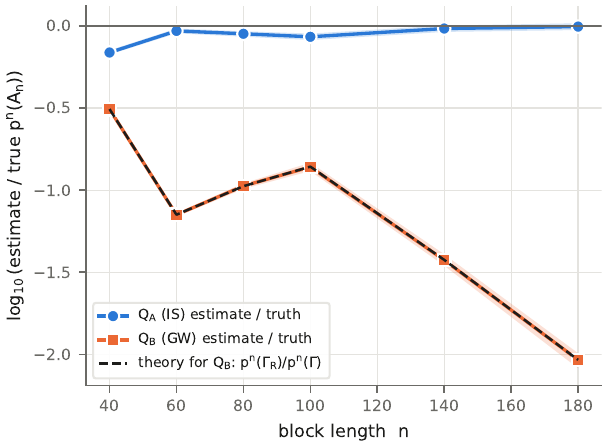}
\caption{IS estimates $\hat p$ of $p^{\otimes n}(A_n)$ ($5000$ samples, $300$ replications):
$\log_{10}(\hat p/p^{\otimes n}(A_n))$, median and $10$--$90\%$ band. $Q_A$ is nearly exact;
At the simulated budget, \(Q_B\) almost never visits \(\Gamma_L\), so its estimate effectively captures only the mass of \(\Gamma_R\). (dashed). The knee at $n=60$ is the
integer rounding of $\Gamma_R$ at finite $n$.}
\label{fig:ex_est}
\end{figure}

\subsection{Constrained Discovery and Symmetry Breaking}In practical applications, the ability to shift the nominal distribution is often limited by a budget $\delta > 0$, requiring the proposal to satisfy $D(Q\|p) \le \delta$. If $\delta < \inf_{Q \in \Gamma} D(Q\|p)$, the rare event $A_n(\Gamma)$ remains rare under any feasible proposal, and the hitting-time exponent $J_{\mathrm{hit}}$ is strictly positive. In this regime, the design problem is often ill-posed under classical criteria. If multiple proposals achieve the same minimum distance to the rare set, the mass-based objective $\min D(Q\|p)$ provides no guidance on which "direction" to take. Guesswork serves as a fundamental lexicographic tie-breaker.

\begin{remark}
For \(p=[0.8,0.2]\), $\Gamma=\{q<0.75\}\cup\{q>0.84704\}$,
and $\delta=0.05$, the boundary of the rare set then contains two types, $Q_2=[0.75,0.25]$, and $Q_1=[0.84704,0.15296]$,
which are exactly equidistant in KL. Classical importance sampling is therefore indifferent. In contrast, the guesswork criterion strictly prefers the right boundary point $Q_1$. Moreover, under the budget constraint the guesswork optimizer moves further into the right component, selecting $Q_{GW}^* \approx [0.89745,0.10255]$.
Thus, rather than merely selecting among KL-equivalent boundary points, guesswork breaks the symmetry and drives the solution toward the type that minimizes search depth within the admissible set.
\end{remark}

We define the optimal constrained discovery proposal as the solution to a lexicographic optimization problem. Let $\mathcal{Q}^*_\delta$ be the set of proposals that minimize the discovery delay $J_{\mathrm{hit}}$ within the budget $\delta$: 
\[
\mathcal{Q}^\delta = \arg \min_{Q : D(Q||p) \le \delta} \left( \inf_{R \in \Gamma} D(R||Q) \right).
\]

The guesswork-refined proposal $Q_{\mathrm{GW}}$ is then chosen to maximize the quality of the discovered trajectory: $Q_{\mathrm{GW}}^* \in \arg\min_{Q \in \mathcal{Q}^\delta} g(Q_{\text{proj}})$,
where $Q_{\text{proj}}$ is the asymptotic type in $\Gamma$ targeted by the proposal $Q$.
\begin{corollary}[Guaranteed Quality Gain]
Let $Y^n_{\mathrm{GW}}$ be the trajectory discovered by $Q_{\mathrm{GW}}^*$ and $Y^n_{\text{IS}}$ be a trajectory discovered by any other proposal in $\mathcal{Q}^*_\delta$. Both trajectories share the same discovery speed, yet the guesswork-guided choice satisfies:
$$\lim_{n\to\infty} \frac{1}{n} \log \frac{G_n(Y^n_{\text{IS}})}{G_n(Y^n_{\mathrm{GW}})} = g(Q_{\text{proj, IS}}) - \min_{Q \in Q^\delta}g(Q_{\text{proj}}) \ge 0.$$
\end{corollary}
This result guarantees that even when discovery speed is prioritized or limited by a budget, the guesswork framework identifies the "entry points" of the rare set.
\section{Conclusion}
While traditional methods in importance sampling focus on estimating the probability of a rare set, many practical applications, such as security audits or stress-testing, require generating a single, "believable" trajectory as quickly as possible. We demonstrated that the fastest way to hit a target set is not necessarily the best way to understand it. We introduced a Quality of Discovery metric based on the principle of minimal surprisal (or description length). We proved that the most meaningful discovery is the trajectory that is easiest to find and describe under the system’s original rules.

This guesswork-guided design principle ensures that a simulation targets easily discoverable and best-explained trajectories of the rare set. This framework provides a rigorous bridge between randomized sampling and systematic search (e.g., list-based guesswork), ensuring that discovered events (for example, failures in a system) are not just fast to find, but are the most representative examples.


\bibliographystyle{IEEEtran}
\bibliography{importance_sampling,guesswork}

\newpage

\appendices
\section{Proof of \Cref{prop:arikan_average}}\label{app:arikan}
Group the expectation by types. If \(\hat P_{x^n}=Q\), then Proposition~\ref{prop:pathwise_guesswork} gives \(G_n(x^n)\doteq e^{n g(Q)}\). Hence, for any \(\rho>0\),
\begin{align*}
\E_{p_n}\bigl[G_n(X^n)^\rho\bigr]
&=
\sum_{Q\in\mathcal P_n(\mathcal X)}\ \sum_{x^n\in T(Q)} p_n(x^n)\, G_n(x^n)^\rho \\
&\doteq
\sum_{Q\in\mathcal P_n(\mathcal X)} \exp\bigl(n[\rho g(Q)-D(Q\|p)]\bigr).
\end{align*}
because \(|T(Q)|\doteq e^{nH(Q)}\) and \(p_n(x^n)=\exp(-n[H(Q)+D(Q\|p)])\) for every \(x^n\in T(Q)\). Since the number of \(n\)-types is subexponential in \(n\),
\[
\lim_{n\to\infty}\frac{1}{n}\log \E_{p_n}\bigl[G_n(X^n)^\rho\bigr]
=
\sup_{Q\in\mathcal P(\mathcal X)} \bigl[\rho g(Q)-D(Q\|p)\bigr].
\]

We wish to show that this is, in fact, Arikan's exponent. To this end, define the cross-entropy level
\[
C(Q):=H(Q)+D(Q\|p),
\]
and let
\[
\phi(c):=\sup_{R\in\mathcal P(\mathcal X):\, C(R)\le c} H(R).
\]
Then \(g(Q)=\phi(C(Q))\). 
\if{0}Therefore, if \(C(Q)=c\),
\[
\rho g(Q)-D(Q\|p)
=
\rho\phi(c)-c+H(Q)
\le
(\rho+1)\phi(c)-c,
\]
since \(H(Q)\le \phi(c)\). Equality is achieved by any maximizer of \(\phi(c)\).

\paragraph{Missing step in Appendix A.}
Recall that
\[
C(Q):=H(Q)+D(Q\|p)
\]
and
\[
\phi(c):=\sup_{R:\, C(R)\le c} H(R).
\]
Then
\[
g(Q)=\phi(C(Q)).
\]
\fi
Therefore
\begin{align*}
\rho g(Q)-D(Q\|p)&=\rho\phi(C(Q))-D(Q\|p)
\\
&=
\rho\phi(C(Q))-C(Q)+H(Q)
\\
&\le 
(\rho+1)\phi(C(Q))-C(Q), 
\end{align*}
where the last inequality is since, by definition of \(\phi\), $H(Q)\le \phi(C(Q))$. As a result, 
\begin{align}
\sup_Q  & \{\rho g(Q)-D(Q\|p)\}  \nonumber
\\
& \le \sup_Q\{(\rho+1)\phi(C(Q))-C(Q)\} \nonumber
\\
& \le \sup_c\{(\rho+1)\phi(c)-c\}\nonumber
\\
& = \sup_c\{ (\rho+1)\sup_{R:\,C(R)\le c}H(R)-c \}\nonumber
\\
& = \sup_c\{ \sup_{R:\,C(R)\le c}\{(\rho+1) H(R)-c\} \}\nonumber
\\
& \le \sup_c\{ \sup_{R:\,C(R)\le c}\{(\rho+1) H(R)-C(R)\} \}\nonumber
\\
& \le \sup_R\{(\rho+1)H(R)-C(R)\} \nonumber
\\
& = \sup_R\{\rho H(R)-D(R\|p)\}. \nonumber
\end{align}

For the reverse inequality, since \(R\) is feasible in the definition of \(\phi(C(R))\), we have
\[
\rho g(R)-D(R\|p)
\ge
\rho H(R)-D(R\|p).
\]
Taking the supremum over \(R\) gives
\[
\sup_Q\{\rho g(Q)-D(Q\|p)\}
\ge
\sup_R\{\rho H(R)-D(R\|p)\},
\]
hence
\begin{equation}\label{two optimization problems}
\sup_{Q\in\mathcal P(\mathcal X)} \bigl[\rho g(Q)-D(Q\|p)\bigr]
=
\sup_{R\in\mathcal P(\mathcal X)} \bigl[\rho H(R)-D(R\|p)\bigr].
\end{equation}

The solution to the optimization in \eqref{two optimization problems} is well known. Define
\[
P_\rho(a):=\frac{p(a)^{1/(1+\rho)}}{\sum_{b\in\mathcal X} p(b)^{1/(1+\rho)}},
\qquad a\in\mathcal X.
\]
A direct calculation yields
\begin{align*}
\rho H(R)-D(R\|p)
&=
(\rho+1)\log \sum_{a\in\mathcal X} p(a)^{1/(1+\rho)} \\
&\qquad - (\rho+1)D(R\|P_\rho).
\end{align*}
Hence the supremum over $R$ is attained at \(R=P_\rho\), and
\begin{align*}
\sup_{R\in\mathcal P(\mathcal X)} \bigl[\rho H(R)-D(R\|p)\bigr]
&=
(\rho+1)\log \sum_{a\in\mathcal X} p(a)^{1/(1+\rho)} \\
&=
\rho H_{1/(1+\rho)}(p).
\end{align*}
This is precisely Ar\i{}kan's exponent. 


\section{Asymptotic Analysis of the Discovered Trajectory}\label{app:discovery_projection}

The main text focuses primarily on proposals of the form
$q_n=Q^{\otimes n}$ with $Q\in\Gamma$. In this case, the
proposal makes the target event typical and the empirical type
of a discovered trajectory concentrates around $Q$.

More generally, however, a proposal $Q$ need not belong to
$\Gamma$. This situation arises, for example, when the admissible
change of measure is constrained. The result below characterizes
both the discovery time and the asymptotic type of the trajectory
returned by the repeated-sampling procedure in this more general
case. In particular, the discovered trajectory concentrates on the
KL projection of the proposal onto the target set.

\begin{theorem}
\label{thm:general_discovery_projection}
Let $p_n=P^{\otimes n}$, $q_n=Q^{\otimes n}$, where $P$ and $Q$ have full support on the finite alphabet
$\mathcal X$. Let
\[
A_n(\Gamma)
=
\left\{
y^n\in\mathcal X^n:
\hat P_{y^n}\in\Gamma
\right\},
\]
where $\Gamma\subseteq\mathcal P(\mathcal X)$ is closed.

Generate
\[
Y_1^n,Y_2^n,\ldots
\stackrel{\mathrm{i.i.d.}}{\sim}
Q^{\otimes n},
\]
and define
\[
\tau_{A_n}
:=
\inf\left\{
t\ge 1:Y_t^n\in A_n(\Gamma)
\right\},
\qquad
Y^n:=Y_{\tau_{A_n}}^n.
\]

Assume that
\[
R_Q^*
=
\arg\min_{R\in\Gamma}D(R\|Q)
\]
is unique, and that $R_Q^*$ can be approached by $n$-types
belonging to $\Gamma$, i.e., there exists
$R_n\in\Gamma\cap\mathcal P_n(\mathcal X)$ such that
$R_n\to R_Q^*$.

Then:

\begin{enumerate}
\item The hitting-time exponent is
\begin{equation}
\label{eq:general_hit_exponent}
\lim_{n\to\infty}
\frac{1}{n}
\log
\mathbb E_{Q^{\otimes n}}[\tau_{A_n}]
=
\inf_{R\in\Gamma}D(R\|Q)
=
D(R_Q^*\|Q).
\end{equation}

\item The empirical type of the returned trajectory satisfies
\begin{equation}
\label{eq:conditional_type_concentration}
\hat P_{Y^n}
\xrightarrow{P}
R_Q^*.
\end{equation}

\item Consequently, the expected normalized nominal surprisal
of the returned trajectory satisfies
\begin{equation}
\label{eq:general_discovery_quality}
\lim_{n\to\infty}
\mathbb E
\left[
-\frac{1}{n}\log P^{\otimes n}(Y^n)
\right]
=
H(R_Q^*)+D(R_Q^*\|P).
\end{equation}
\end{enumerate}
\end{theorem}

\begin{proof}
Let
\[
\pi_n
:=
Q^{\otimes n}(A_n(\Gamma)).
\]
Since each trial is independent, $\tau_{A_n}$ is geometrically
distributed with success probability $\pi_n$, and hence
\[
\mathbb E[\tau_{A_n}]
=
\frac{1}{\pi_n}.
\]

Using the method of types,
\[
\pi_n
=
\sum_{R\in\Gamma\cap\mathcal P_n(\mathcal X)}
Q^{\otimes n}(T(R)).
\]
For every $n$-type $R$,
\begin{multline}
(n+1)^{-|\mathcal X|}
\exp\{-nD(R\|Q)\}
\le
Q^{\otimes n}(T(R)) \\
\le
\exp\{-nD(R\|Q)\}.
\end{multline}
Since the number of $n$-types is polynomial in $n$, these
bounds imply
\[
-\lim_{n\to\infty}
\frac{1}{n}\log \pi_n
=
\inf_{R\in\Gamma}D(R\|Q)
=
D(R_Q^*\|Q).
\]
Therefore,
\[
\lim_{n\to\infty}
\frac{1}{n}
\log\mathbb E[\tau_{A_n}]
=
D(R_Q^*\|Q),
\]
which proves \eqref{eq:general_hit_exponent}.

We next characterize the empirical type of the returned
trajectory. First note that the stopping procedure induces exactly
the conditional distribution
\begin{equation}
\label{eq:return_conditional_law}
\Pr\{Y^n=y^n\}
=
\frac{Q^{\otimes n}(y^n)}
     {Q^{\otimes n}(A_n(\Gamma))}
\mathbf 1\{y^n\in A_n(\Gamma)\}.
\end{equation}
Indeed,
\begin{align*}
\Pr\{Y^n=y^n\}
&=
\sum_{t=1}^{\infty}
\Pr\{\tau_{A_n}=t,Y_t^n=y^n\} \\
&=
\sum_{t=1}^{\infty}
(1-\pi_n)^{t-1}
Q^{\otimes n}(y^n)
\mathbf 1\{y^n\in A_n(\Gamma)\} \\
&=
\frac{Q^{\otimes n}(y^n)}
     {\pi_n}
\mathbf 1\{y^n\in A_n(\Gamma)\}.
\end{align*}

Fix $\epsilon>0$ and define
\[
\Gamma_\epsilon
:=
\left\{
R\in\Gamma:
\|R-R_Q^*\|_1\ge\epsilon
\right\}.
\]
Because $\Gamma$ is closed, the probability simplex is compact,
and $R_Q^*$ is the unique minimizer of $D(R\|Q)$ over
$\Gamma$, there exists $\Delta_\epsilon>0$
such that
\[
\inf_{R\in\Gamma_\epsilon}D(R\|Q)
\ge
D(R_Q^*\|Q)+\Delta_\epsilon.
\]

It follows from the same type-class bounds that
\begin{align*}
&
\Pr\left\{
\|\hat P_{Y^n}-R_Q^*\|_1\ge\epsilon
\right\}
\\
&\qquad =
\frac{
Q^{\otimes n}
\left(
\hat P_{Y^n}\in\Gamma_\epsilon
\right)
}{
Q^{\otimes n}
\left(
\hat P_{Y^n}\in\Gamma
\right)
}
\\
&\qquad \le
\operatorname{poly}(n)
\exp\left\{
-n(\Delta_\epsilon-o(1))
\right\},
\end{align*}
which tends to zero. Hence $\hat P_{Y^n}
\xrightarrow{P}
R_Q^*$, proving \eqref{eq:conditional_type_concentration}.

Finally, for every sequence $y^n$ with empirical type
$\hat P_{y^n}=R$,
\begin{align}
-\frac{1}{n}\log P^{\otimes n}(y^n)
&=
-\sum_{x\in\mathcal X}R(x)\log P(x)
\nonumber\\
&=
H(R)+D(R\|P).
\label{eq:surprisal_type_identity}
\end{align}
Since $P$ has full support,
\[
C_P(R)
:=
H(R)+D(R\|P)
=
-\sum_x R(x)\log P(x)
\]
is continuous and bounded on $\mathcal P(\mathcal X)$.
Combining \eqref{eq:conditional_type_concentration} with
\eqref{eq:surprisal_type_identity} therefore gives
\[
-\frac{1}{n}\log P^{\otimes n}(Y^n)
\xrightarrow{P}
H(R_Q^*)+D(R_Q^*\|P).
\]
Boundedness then implies convergence of expectations, yielding
\eqref{eq:general_discovery_quality}.
\end{proof}

\begin{remark}[Connection to guesswork-guided discovery]
\label{rem:projection_connection}
Theorem~\ref{thm:general_discovery_projection} separates the
effect of a proposal $Q$ into two quantities. Its discovery speed
is governed by $D(R_Q^*\|Q)$, whereas the quality of the trajectory eventually discovered is
governed by the KL projection $R_Q^*
\in
\arg\min_{R\in\Gamma}D(R\|Q)$. Thus the proposal $Q$ determines which part of the rare set is
typically encountered when the event is finally observed.

In the unconstrained setting considered in
Theorem~\ref{main theorem unconstrained}, the proposed
distribution satisfies $Q_{\mathrm{GW}}^*\in\Gamma$. Therefore $R_{Q_{\mathrm{GW}}^*}^*
= Q_{\mathrm{GW}}^*$, and $D(R_{Q_{\mathrm{GW}}^*}^* \|Q_{\mathrm{GW}}^*)=0$.
Hence the event is discovered in subexponential time, and the
empirical type of the returned trajectory concentrates around
$Q_{\mathrm{GW}}^*$.

More generally, when constraints prevent the proposal itself
from belonging to $\Gamma$, the quantity $R_Q^*$ is precisely
the asymptotic target type, denoted by $Q_{\mathrm{proj}}$ in
the constrained-discovery formulation of the main text.
\end{remark}

\begin{corollary}[Proposal inside the target set]
\label{cor:q_in_gamma}
Under the assumptions of
Theorem~\ref{thm:general_discovery_projection}, if $Q\in\Gamma$,
then $R_Q^*=Q$. Consequently,
\[
J_{\mathrm{hit}}(Q^{\otimes n})=0,
\qquad
\hat P_{Y^n}\xrightarrow{P}Q,
\]
and
\[
\lim_{n\to\infty}
\mathbb E
\left[
-\frac{1}{n}\log P^{\otimes n}(Y^n)
\right]
=
H(Q)+D(Q\|P).
\]
\end{corollary}

\end{document}